\documentclass[12pt,journal,draftcls,onecolumn]{IEEEtran}

\usepackage{pstricks}
\usepackage{amsthm}
\usepackage{enumerate}
\usepackage{tikz}
\usepackage{tikzscale}
\usepackage{verbatim}

\usepackage{graphicx}
\usepackage{amsfonts}
\usepackage{algorithmic}
\usepackage{algorithm}
\usepackage{amssymb}
\usepackage{bbm}
\usepackage{color}
\usepackage{cite}
\usepackage{amsfonts}
\usepackage{relsize}
\usepackage[cmex10]{amsmath}
\usepackage{graphicx}
\usepackage{array}
\usepackage{amssymb}
\usepackage{bbm}
\usepackage{mdwmath}
\usepackage{mdwtab}
\usepackage{fixltx2e}
\usepackage{url}
\usepackage[center]{caption}
\usepackage[font+=scriptsize]{subcaption}
\newtheorem{theorem}{Theorem}
\newtheorem{property}{Property}
\newtheorem{lemma}{Lemma}

\newtheorem{remark}{Remark}

\def\floor#1{\lfloor #1 \rfloor}

\usepackage{amsmath,amssymb,mathrsfs}
\usepackage{epsfig,epstopdf,epsf,graphicx,graphics}
\usepackage{url}
\usepackage{color}
\usepackage{mathtools}

\newcommand\numberthis{\addtocounter{equation}{1}\tag{\theequation}}

\newcommand{\C}{\msf{C}}

\newcommand{\mcal}{\mathcal}

\newcommand{\mb}{\mathbf}
\newcommand{\mbb}{\mathbb}
\newcommand{\msf}{\mathsf}

\newcommand{\RN}[1]{%
      \textup{\uppercase\expandafter{\romannumeral#1}}%
  }

\input{widebar_hack.tex}

\allowdisplaybreaks
\IEEEoverridecommandlockouts

\title{Wireless Network Simplification:\\The Performance of Routing}

\author{
    \IEEEauthorblockN{Yahya H. Ezzeldin, Ayan Sengupta, Christina Fragouli} 
\thanks{
    Y.~H.~Ezzeldin and C.~Fragouli are with the Electrical Engineering Department at the University of California, Los Angeles, CA 90095 USA (e-mail: \{yahya.ezzeldin, christina.fragouli\}@ucla.edu). 
    The research carried out at UCLA was partially funded by NSF under award number 1514531. A. Sengupta is with the Electrical Engineering Department at Stanford University, Stanford, CA 94305 (e-mail: ayans@stanford.edu) and is supported by SNSF Early Postdoc Mobility Fellowship.  

The results in this paper were presented in part at the 2016 IEEE International Symposium on Information Theory.
}

}
\begin{document}
\maketitle
\begin{abstract}
  Consider a wireless Gaussian network where a source wishes to communicate with a destination with the help of N full-duplex  relay nodes.  Most practical systems today route information from the source to the destination using the best path that connects them. In this paper, we show that routing can in the worst case result in an unbounded gap from the network capacity - or reversely, physical layer cooperation can offer unbounded gains over routing.  More specifically, we show that for  $N$-relay Gaussian networks with an arbitrary topology, routing can in the worst case guarantee an approximate fraction $\frac{1}{\floor{N/2} + 1}$ of the capacity of the full network, independently of the SNR regime.  We prove that this guarantee is fundamental, i.e., it is the highest worst-case guarantee that we can provide for routing in relay networks.   Next, we consider how these guarantees are refined for  Gaussian layered relay networks with $L$ layers and $N_L$ relays per layer.   We prove that for arbitrary $L$ and $N_L$, there always exists a route in the network that approximately achieves at least $\frac{2}{(L-1)N_L + 4}$ $\left(\mbox{resp.}\frac{2}{LN_L+2}\right)$ of the network capacity for odd $L$ (resp. even $L$), and there exist networks where the best routes exactly achieve these fractions.
 These results are formulated within the network simplification framework, that asks what fraction of the capacity we can achieve by using a subnetwork (in our case, a single path). A fundamental step in our proof is a simplification result for  MIMO antenna selection that may also be of independent interest. To the best of our knowledge, this is the first result that characterizes, for general wireless network topologies, what is the performance of routing with respect to physical layer cooperation techniques that approximately achieve the network capacity.

\end{abstract}

\IEEEpeerreviewmaketitle
\section{Introduction} \label{sec:intro}

Consider a wireless Gaussian network where a source wishes to communicate with a destination using the help of wireless full-duplex relay nodes.
Work in information theory has shown that we can approximately achieve the network capacity by using physical layer cooperation schemes \cite{ADT2011,NNC}; there has also been an increasing interest in the community to translate these schemes towards practical networks \cite{duarte2013quantize,brahma2014quilt,sengupta2016consistency}. 

Currently, the widespread approach in practical networks is to route information from the
source to the destination through a single multi-hop path, consisting of successive point-to-point transmissions. {Routing} is considered an appealing option since it has low complexity, provides  energy savings (by only powering the network nodes belonging to the selected route), and creates limited network interference - as a result, there is a rich literature on how to use routing for wireless network applications \cite{AODV}, \cite{DSR},\cite{OLSR},\cite{DSDV}. However, even if we select to route along the best (highest capacity) path that connects the source to the destination, we could be significantly  under-utilizing the available network capacity.

In this paper we ask, given an arbitrary wireless network, how does the capacity of the best path (achieved by routing) compare to the network capacity achieved by optimally using physical layer cooperation over all the available relays.
Answering this question can help motivate the use of physical layer cooperation,
by better understanding where we expect significant benefits as compared to routing. Moreover, this work offers new results  within the network simplification framework, that asks what fraction of the capacity we can achieve by using a subnetwork (in our case, a single path).


We prove that routing can in the worst case result in a fraction of the network capacity that decreases with the number of nodes in the network, independently of SNR -  or reversely, physical
layer cooperation can offer gains over routing that grow linearly with the number of nodes in the network. 
In particular, we prove that for any $N$-relay Gaussian network, there always exists a route in the networks that achieves $\frac{1}{\floor{N/2}+1}$ of the approximate capacity of the full network. 
Moreover, we provide networks for which this fraction guarantee is tight, proving that the bound of $\frac{1}{\floor{N/2} +1}$ is the best worst case fraction guarantee for the achievable rate of routing.
This is a surprising result when put in contrast with the result in \cite{NOF_Simplification} which shows that, if we select the best route over a diamond N-relay network, we can always approximately achieve $\frac{1}{2}$ of the network capacity, independently of the number $N$ of relay nodes.
This suggests that the independence of the number of nodes in the guarantee might be a property of the diamond network.

To further understand this, we consider the capacity fraction guarantee when
we have a layered Gaussian relay network with $L$ layers and $N_L$ relays per layer
(the diamond network corresponds to the case of a single layer).
       We prove that there always exists a path in the network (by selecting one relay per layer) that can achieve $\frac{2}{(L-1)N_L + 4}$ (resp. $\frac{2}{LN_L+2}$) of the approximate capacity of the full network for odd $L$ (resp. even $L$).        We also prove that there exist networks where the best path achieves this bound.         This  result, refined for layered networks, admits the result in \cite{NOF_Simplification}  as a special case by setting $L=1$. The main intuition is that for $L\geq2$ subsequent layers act as MIMO channels, an effect not captured for $L=1$. 

    At the heart of our proofs, we deal with the problem of analyzing how subsets (in terms of antennas) of a MIMO channel behave with respect to the entire MIMO channel. 
        We therefore prove the following subsidiary simplification result on the MIMO channel which might be of independent interest: for every $n_t \times n_r$ Gaussian MIMO channel with i.i.d inputs, the best $k_t \times k_r$ subchannel approximately achieves a fraction $\frac{\min\{k_t, k_r\}}{\min\{n_t, n_r\}}$ of the full MIMO capacity, universally over all channel configurations.

\subsection{Related Work}
For the Gaussian full-duplex relay networks, the capacity is not known in general. The tightest known universal upper bound for the capacity is the information theoretic cut-set upper bound. In \cite{ADT2011}, the authors showed that the network can achieve a rate that is a constant gap away from the cut-set upper bound through the quantize-map-and-forward relaying strategy. 
Similar strategies \cite{NNC}, \cite{DDF} have been shown to achieve a similar result. 
For all these strategies, the gap is linear in the number of nodes $N$ in the network.  Although for several network topologies, it has been shown that the gap is sublinear \cite{KOEG_logGap, SWF_logGap}, \cite{CO_logGap}. 
It has been shown through the results in \cite{CO_sublineargap} and \cite{WO2015}, that a linear gap to the cut-set bound is indeed fundamental for the class of Gaussian full-duplex relay networks.

In the thread of work on wireless network simplification, \cite{NOF_Simplification} studied the problem for the Gaussian full-duplex diamond network. The authors in \cite{NOF_Simplification} provided universal capacity guarantees for selecting $k$-relay subnetworks, where they show that selecting $k$ out of $N$ relays in the network is always guaranteed to achieve at least $\frac{k}{k+1}$ of the full network capacity, to within a constant gap.  When applied to a single route selection, \cite{NOF_Simplification} states that a route in a diamond relay network can always approximately achieve half of the capacity of the capacity of the diamond network. The work in \cite{Javad_Simplification} extended the result in \cite{NOF_Simplification} for some scenarios of the Gaussian FD diamond network with multiple antennas at the nodes. 
The network simplification problem have also been studied recently in \cite{CEFT_submodularity} for Gaussian half-duplex diamond relay networks, where the authors showed we can always select $N-1$ relays and approximately achieve $\frac{N-1}{N}$ of the Gaussian half-duplex relay network capacity.
As a scheme-specific performance guarantee (as opposed to guaranteeing capacity fractions), the work of \cite{agnihotri12} proved upper bounds on multiplicative and additive gaps for relay selection based on the amplify-and-forward scheme, primarily for diamond full-duplex networks. 
In \cite{agrawal2016}, the authors characterized the performance of network simplification (in terms of achievable secure capacity) for layered full-duplex relay networks operating using amplify-and-forward  in the presence of an eavesdropper.

Another thread of related work pertains to algorithm design for finding near-optimal subnetworks. \cite{BSF_Selection} and \cite{Ozgur_Selection} made progress in that direction, by providing low-complexity heuristic algorithms for near-optimal relay selection. The work of \cite{ bletsas2006}, \cite{tannious08} proposed algorithms for only selecting the best route (in terms of cooperative diversity) in one-layer networks. As far as we know, this is the first work that theoretical proves worst case performance guarantees for the capacity of the best path over an arbitrary wireless network.

\subsection{Paper Organization}
The paper is organized as follows.
Section~\ref{sec:model} describes the $N$-relay Gaussian FD network and its approximate capacity expression.
Section~\ref{sec:model} also introduces notation  that will be in the remainder of the paper.
Section~\ref{sec:main_results} discusses the main results in the paper and compares th different guarantees on the achievable rate.
Section~\ref{sec:main_results_mimo} derives a simplification result for MIMO channels with i.i.d inputs which is a key ingredient in the proof of our main results.
Section~\ref{sec:main_results_general} proves the universal guarantee on the achievable rate by the best route in a Gaussian FD relay network, in terms of a fraction of the full network capacity.
In section~\ref{sec:main_results_layered}, a refined guarantee for Gaussian FD layered networks is proved.
Section~\ref{sec:conclusion} concludes the discussion in the paper. 
Some parts of the proofs are delegated to the Appendices.

\section{System Model and Preliminaries} \label{sec:model}
Throughout the paper, we denote with $[a\!:\!b]$ the set of integers from $a$ to $b$, where $b\geq a$.
We consider a Gaussian relay network where the \textit{Source} ($S$) wants to communicate with the \textit{Destination} ($D$) through the help of $N$ relays operating in full-duplex.
The set of all nodes in the network is denoted by $\mcal{V}$.
Nodes in $\mcal{V}$ are indexed with the integers $[0\!:\!N+1]$ where the Source and Destination are indexed by 0 and $N+1$, respectively.

At any time $t$, the received signal $Y_j[t]$ at node $j$ is a function of the transmitted signals from all other nodes in the network (except $D$),
\begin{align}
    \label{ntwk_model}
    Y_j[t] = \sum_{\substack{i=0,\\i \neq j}}^{N} h_{ij} X_i[t] + W_{j}[t], \quad \forall j \in [1:N+1],
\end{align}
where: (i) $X_i$ is the transmitted signal from the $i$-th node; (ii) the additive white Gaussian noise $W_{j} \sim \mcal{CN}(0,1)$ at $j$-th node is independent of the inputs, as well as of the noise terms at the other nodes; (iii)
the (complex) channel gain between nodes $i$ and $j$ is denoted by $h_{ij} \in \mbb{C}$. Transmitted signals from each network node satisfy an average power constraint $\mathbb{E}[|X_i|^2]\leq 1 \quad \forall i \in [0:N]$.
%

The exact capacity $\C$ of the network described in \eqref{ntwk_model} is not known in general. However, in \cite{ADT2011} the authors prove that it is within a constant gap\footnote{By constant gap, we refer to terms that are independent of the channel coefficients in the network.} from the cutset upper bound evaluated with i.i.d Gaussian input distributions, given by
\begin{equation}\label{eq:C_bar_expression}
    \widebar{\C} \triangleq \min_{\Omega \in 2^{\mcal{V}}} \widebar{\C}(\Omega,\mcal{V}),
\end{equation}
  {where}
  \begin{equation}
  \label{eq:cut_value}
    \widebar{\C}(\Omega,\mcal{V}) \triangleq \log\text{det}\left( \mb{I} + \mb{H}_\Omega {\mb{H}_\Omega}^\dagger \right). \numberthis
\end{equation}
The matrix $\mb{H}_\Omega$ represents a MIMO channel matrix from transmitting nodes in $\Omega$ to receiving nodes in $\Omega^c = \mcal{V}\backslash \Omega$. 
We refer to $\Omega \subseteq \mcal{V}$ as a \emph{``cut''} in the network. 
In the rest of the paper, we work with the approximate capacity $\widebar{\C}$ in place of the network capacity to prove our results.

In a $N$-relay Gaussian network, we denote the capacity of the point-to-point channel between node $i$ and node $j$ as
\begin{align*}
    R_{i\to j} \triangleq \log\left(1 + |h_{ij} |^2\right),\qquad \forall i,j \in[0:N+1].
\end{align*}
A path (route) $\mcal{P}$ in an $N$-relay Gaussian FD network is defined by a sequence of $|\mcal{P}|+1$ non-repeating nodes $\{v_0,v_1,\dots,v_{|\mcal{P}|}\}$, where $v_0 = 0$, $v_{|\mcal{P}|} = N+1$ and $v_i \in [1:N], \forall i \in [1\!:\!|\mcal{P}|-1]$.
The path $\mcal{P}$, therefore, defines a line network from $S$ to $D$ induced by the links connecting nodes $v_{i-1}$ and $v_{i}$ for $i \in [1:|\mcal{P}|]$. 
The capacity of the path $\mcal{P}$ is denoted by $\C_{\mcal{P}}$ and is known to be equal to 
\begin{align}
    \label{eq:cap_path}
    \C_{\mcal{P}} = \min_{0 \leq i \leq | \mcal{P} |-1} R_{v_i\to v_{i+1}},
\end{align}
and can be achieved through the Decode-And-Forward scheme \cite{NIT_book}.

\section{Main Results}\label{sec:main_results}
The main results of this paper are summarized in the following theorems.

\begin{theorem}
    In any $N$-relay Gaussian network with approximate capacity $\widebar{\C}$, there exists a path $\mcal{P}$ (line network) such that the  capacity $\C_\mcal{P}$ of the path satisfies
    \begin{equation}\label{thm:single_path_low_bound}
        \C_\mcal{P} \geq \frac{1}{\floor{N / 2} + 1} \widebar{\C} - 2\log\left(\frac{N+2}{2}\right). 
    \end{equation}
    Moreover, there exists a class of networks with $N$ relays such that for all paths $\mcal{P}$,
    \begin{align}\label{thm:single_path_tight_bound}
        \C_\mcal{P} \leq \frac{1}{\floor{N / 2} + 1} \widebar{\C}.
    \end{align}
    \label{thm:single_path}
\end{theorem}
Theorem~\ref{thm:single_path} states that for every $N$-relay Gaussian network, the capacity of the best route $\mcal{P}$ in the network is guaranteed to at least be a fraction $2/(N+2)$ of the approximate capacity of the full network. 
The theorem also states that the fraction is tight, that is, for some wireless networks, the best route cannot achieve a capacity greater than $2/(N+2)$ of the approximate capacity of the full network. 
Thus for an optimal routing protocol, the guaranteed rate achieved through routing (in comparison to schemes that use physical layer cooperation across the network nodes) grows inverse proportionally as the number of nodes in the network increase.
\begin{remark}{\rm
    Although the result in Theorem \ref{thm:single_path} is true in general, 
we could get alternative characterizations if we are  interested in specific  classes of network topologies.
For example, consider an $N$-relay Gaussian FD diamond network. The result in Theorem \ref{thm:single_path} states that the guarantee we can give on the capacity of the best route is $\frac{1}{\floor{N/2}+1} \widebar{\C}$.
However, the result in \cite[Theorem 1]{NOF_Simplification} proves that in diamond networks, a route is guaranteed to achieve at least $\frac{1}{2} \widebar{\C}$.
In other words, for this particular case, the guarantee is independent of number of relays $N$ unlike the guarantee in Theorem~\ref{thm:single_path} above.
This suggests that the bound can be refined if we restrict ourselves to a class of $N$-relay Gaussian networks with a specific topology.
The following theorem explores this for the class of layered networks, which also includes diamond networks.
}
\end{remark}

\begin{theorem}
    In any $N$-relay Gaussian layered network with $L$ relay layers, $N_L = N/L$ relays per layer and approximate capacity $\widebar{\C}$, there exists a path $\mcal{P}$ (line network) such that the capacity $C_\mcal{P}$ of the path satisfies
        \begin{align}
            \C_\mcal{P} \geq \begin{cases}
            \dfrac{2}{(L-1)N_L + 4} \widebar{\C} - 2  \log(N_L)  , & \quad L\ \text{odd} \\
                \dfrac{2}{LN_L+2} \widebar{\C} - 2  \log(N_L) , & \quad L\ \text{even}.
            \end{cases}
        \end{align}
    Moreover, there exist layered networks with $L$ layers and $N_L$ relays such that for all paths $\mcal{P}$,
    \begin{align}
        \C_\mcal{P} \leq \begin{cases}
                \dfrac{2}{(L-1)N_L + 4} \widebar{\C},& \quad L\ \text{odd} \\
                \dfrac{2}{LN_L+2} \widebar{\C}, & \quad L\ \text{even}.
            \end{cases}
            \label{eq:single_path_layered_ntwk}
        \end{align}
        \label{thm:single_path_layered}
\end{theorem}
\begin{remark}{\rm
        By inspecting the bounds in Theorem~\ref{thm:single_path} and Theorem~\ref{thm:single_path_layered}, it is not hard to see that when the number of layers $L$ is even, the fractions in both theorems coincide. 
        To highlight the refinement in the bound of Theorem~\ref{thm:single_path_layered}, consider an example network with $L = 3$ and $N_L = 10$ ($N = L\times N_L = 30$). 
  Theorem~\ref{thm:single_path} guarantees that there exists a route that can achieve a fraction $1/16$ of $\widebar{\C}$. 
  On the other hand, Theorem~\ref{thm:single_path_layered} presents a higher fraction guarantee of $1/12$. 
  If the structure was changed (with the same number of nodes) so that $L = 6$, $N_L = 5$, then the two bounds coincide. 
  Thus, even with the same number of nodes, the number of transmission stages in the network (and the number of nodes per stage) affect the worst-case performance of a route in comparison to the approximate capacity $\widebar{\C}$.
}
\end{remark}
\begin{remark}{\rm
From Theorem \ref{thm:single_path_layered}, we note that for a diamond network (i.e., $L$ = 1) with approximate capacity $\widebar{\C}_{dia}$, the theorem states that there exists a path $\mcal{P}$ such that $\C_\mcal{P} \geq \frac{1}{2} \widebar{\C}_{dia} - \frac{4N}{N+2}\log(N)$, which is consistent with the result proved in \cite{NOF_Simplification} (with a slightly different gap). 
}
\end{remark}
    {\rm
        Theorem~\ref{thm:single_path_layered} highlights that the fraction being independent of the number of nodes is a unique property of diamond networks (among the class of layered networks). 
        {Intuitively, this unique property of diamond networks can be recognized by studying the structure of cuts in layered networks.
           Approximate capacity cuts in a Gaussian network are represented by  MIMO channels with i.i.d inputs.
           As the number of relay layers $L$ increase (with $N_L > 1$), the minimum dimension of the MIMO channel that represents a cut increases with both $L$ and $N_L$. 
    In the special case of the diamond network ($L=1$), a cut can be represented by at most two orthogonal MIMO channels (a SIMO channel and/or a MISO channel) each of minimum dimension of 1 (thus does not scale with $L$ or $N_L$). 
    For a path, we can see through \eqref{eq:cap_path} that a minimum cut can be represented by a SISO channel independent of the number of nodes in the path. 
    Informally, we can think of the dimensions of these MIMO channels as an indicator of how much information a network can convey from the Source to the Destination. 
    Thus, a path conveys information through a dimension of 1.
    With this view in mind, it is not hard to credit the difference in guarantees to the fact that in a diamond network the minimum cut may be a cross cut of dimension 2 that the best path crosses once, while in a general layered network the minimum cut may be a cross-cut of dimension that grows with $L$ and $N_L$. 
A formal characterization of the aforementioned relationship for MIMO channels with i.i.d inputs based on their dimensions is presented in Theorem \ref{thm:mimo_theorem}. }
}

\begin{theorem}
    \label{thm:mimo_theorem}
    For an $n_t \times n_r$ Gaussian MIMO channel with i.i.d inputs and capacity $\C_{n_t,n_r}$, the best $k_t \times k_r$ subchannel has a capacity $\C_{k_t,r_r}^\star$ such that
    \begin{align}
        \C_{k_t,k_r}^\star \geq \frac{\min(k_t,k_r)}{\min(n_t, n_r)}\ \C_{n_t,n_r} - \log\left({n_t \choose k_t} {n_r \choose k_r}\right).
    \end{align}
    Moreover, this bound is tight up to a constant gap, i.e., there exist $n_t \times n_r$ channels for which
    \begin{align*}
        \C_{k_t,k_r}^\star \leq \frac{\min(k_t,k_r)}{\min(n_t, n_r)}\ \C_{n_t,n_r}.
    \end{align*}
\end{theorem}

\begin{remark}{\rm
        Although the result in Theorem~\ref{thm:mimo_theorem} plays a fundamental role in our proofs of Theorem~\ref{thm:single_path} and Theorem~\ref{thm:single_path_layered}, it is of independent interest for the selection of transmit/recieve antennas in a MIMO channel. 
        The ratio in Theorem \ref{thm:mimo_theorem} is the same one would expect between the maximum multiplexing gains of an $n_t \times n_r$ MIMO channel and its best $k_t \times k_r$ MIMO subchannel at high SNR.
        The difference asserted by Theorem \ref{thm:mimo_theorem} is that the same ratio is also true in MIMO channels with i.i.d inputs for lower SNR levels with a gap that is not a function of SNR.
}
\end{remark}

\section{A Simplification Result for MIMO Channels}
\label{sec:main_results_mimo}
In this section, we derive the result in Theorem~\ref{thm:mimo_theorem} which forms the basis of the proofs of Theorem~\ref{thm:single_path} and Theorem~\ref{thm:single_path_layered}.
Towards proving Theorem~\ref{thm:mimo_theorem}, we first prove two subsidiary results which are employed in our proof. 

Our first subsidiary result proves an incremental version of Theorem~\ref{thm:mimo_theorem} where we only wish to reduce the number of receiving antennas. This is summarized in the following Lemma.
\begin{lemma}
    \label{lem:incremental_mimo}
    For an $n_t \times n_r$ Gaussian MIMO channel with i.i.d inputs and capacity $\C_{n_t,n_r}$ where $n_t \leq n_r$, the best $n_t \times k_r$ MIMO subchannel has a capacity $\C^\star_{n_t,k_r}$ such that

     \begin{subequations}
    \label{eq:combined_cases_const_thm}
     1) For $k_r \leq n_t \leq n_r$, 
     \begin{align}
         \label{k_r_smaller_n_t}
             \C_{k_t,k_r}^\star \geq \frac{k_r}{n_t}\ \C_{n_t,n_r} - \log\left(\frac{{n_r \choose k_r}}{{n_t \choose k_r}}\right).
         \end{align}
     \ \ \ 2) For $n_t \leq k_r \leq n_r$, 
         \begin{equation}
             \label{k_r_larger_n_t}
             \C_{n_t,n_r} \geq\  \C^\star_{n_t,k_r} \geq\ \C_{n_t,n_r} - \log\left(\frac{{n_r \choose k_r}}{{n_r-n_t \choose k_r-n_t}}\right).
         \end{equation}
     \end{subequations}
 \end{lemma}
 \begin{proof} 
     The proof relies on properties of principal submatrices of a Hermitian matrix.
     The detailed proof can be found in Appendix~\ref{appendix:mimo_proof}.
 \end{proof}
We can combine the lower bounds in \eqref{k_r_smaller_n_t} and \eqref{k_r_larger_n_t} as
\begin{align}
    \C^\star_{n_t,k_r} \geq \frac{\min(k_r,n_t)}{n_t}\ \C_{n_t,n_r} - G,
\end{align}
where $G$ is the constant incurred in \eqref{k_r_smaller_n_t} (resp. \eqref{k_r_larger_n_t}) when $k_r \leq n_t$ (resp. $k_r > n_t$).
\begin{remark}{\rm
        Lemma~\ref{lem:incremental_mimo} can also apply to the case where $n_t \geq n_r$ and we wish to select a subchannel $k_t \times n_r$. This can be done by considering the reciprocal MIMO channel or appealing to Sylvester's determinant identity.
}
\end{remark}
 Our second subsidiary result, stated in the following Lemma, provides a guarantee on selecting MIMO subchannels (similar to the statement of Theorem~\ref{thm:mimo_theorem}) without a constant gap.
\begin{lemma}
     \label{zero_gap_thm}
     For an $n_t \times n_r$ Gaussian MIMO channel with i.i.d inputs and capacity $C_{n_t,n_r}$, the best $k_t \times k_r$ MIMO subchannel has a capacity $C^\star_{k_t,k_r}$ such that
     \begin{equation}
                \C_{k_t,k_r}^\star \geq \dfrac{k_t\cdot k_r}{n_t\cdot n_r}\ \C_{n_t,n_r}.
         \label{zero_gap}
     \end{equation}
     Moreover, there exist MIMO channel configurations with i.i.d inputs such that the capacity of the best $k_t \times k_r$ MIMO subchannel is  $\C_{k_t,k_r}^{\star} = \frac{k_t\cdot k_r}{n_t\cdot n_r}\ \C_{n_t,n_r}$.
 \end{lemma}
 \begin{proof}
     Lemma~\ref{zero_gap_thm} is proved in Appendix \ref{sec:proof_zero_gap}.\\
 \end{proof}
\subsection{Proof of Lower Bound in Theorem \ref{thm:mimo_theorem}}
In this subsection we derive the bound on $\C^\star_{k_t,k_r}$ in Theorem~\ref{thm:mimo_theorem} for any chosen dimension $(k_t,k_r)$ using Lemma~\ref{zero_gap_thm} and Lemma~\ref{lem:incremental_mimo}.
Assuming that $n_t \leq n_r$, the proof roughly goes as follows:
From the $n_t \times n_r$ channel, we can create 
an $n_t \times k_r$ subchannel such that $\C^\star_{n_t,k_r} \geq \frac{\min(k_r,n_t)}{n_t}\ \C_{n_t,n_r} - G_1$, by keeping only the best $k_r$ receiver antennas;
from this $n_t \times k_r$ channel, we can next get a $k_t \times k_r$ subchannel such that
\begin{align*}
    \C^\star_{k_t,k_r} &\geq \frac{\min(k_t,k_r)}{\min(n_t,k_r)}\ \C^\star_{n_t,k_r} - G_2 \\
    &\geq \frac{\min(k_t,k_r)}{\min(n_t,n_r)}\ \C_{n_t,n_r} - G_1 - G_2.
\end{align*}
Formally, the constants $G_1$ and $G_2$ and the applications of Lemma~\ref{lem:incremental_mimo} and Lemma{\ref{zero_gap_thm} are captured in the following three cases:
\begin{enumerate}
    \item For $k_t \leq k_r \leq n_t \leq n_r$:
        \begin{align*}
        \C^\star_{k_t,k_r} \stackrel{(a)}\geq& \frac{k_t}{k_r}\ \C^\star_{n_t,k_r} - \log\left(\frac{ {n_t \choose k_t } }{ {k_r \choose k_t } } \right)\\
                \stackrel{(b)}\geq& \frac{k_r}{n_t}\frac{k_t}{k_r} \C_{n_t,n_r} - \frac{k_t}{k_r}\log\left(\frac{ {n_r \choose k_r } }{ {n_t \choose k_r } } \right) - \log\left(\frac{ {n_t \choose k_t } }{ {k_r \choose k_t } } \right)\\
                \geq& \frac{k_t}{n_t}\ \C_{n_t,n_r} - \log\left( {n_t \choose k_t }  \right) -\log\left( {n_r \choose k_r } \right),
        \end{align*}
        where: (a) follows by applying \eqref{k_r_smaller_n_t} on the reciprocal of the $n_t \times k_r$ MIMO channel; (b) applies \eqref{k_r_smaller_n_t} to relate  $\C^\star_{n_t,k_r}$ to $\C_{n_t,n_r}$.\\
    \item For $k_r \leq k_t \leq n_t \leq n_r$:
\begin{align*}
                \C^\star_{k_t,k_r} \stackrel{(c)}\geq&\ \C^\star_{n_t,k_r} - \log\left(\frac{ {n_t \choose k_t } }{ {n_t-k_r \choose k_t-k_r } } \right)\\
                \stackrel{(d)}\geq& \frac{k_r}{n_t}\ \C_{n_t,n_r} - \log\left(\frac{ {n_t \choose k_t } }{ {n_t-k_r \choose k_t-k_r } } \right) -\log\left(\frac{ {n_r \choose k_r } }{ {n_t \choose k_r } } \right)\\
                \geq& \frac{k_r}{n_t}\ \C_{n_t,n_r} - \log\left( {n_t \choose k_t }  \right) -\log\left( {n_r \choose k_r } \right),
\end{align*}
where: (c) relates $\C_{k_t,k_r}$ to $\C_{n_t,k_r}$ using \eqref{k_r_larger_n_t}; relation (d) follows by applying \eqref{k_r_smaller_n_t} to the $n_t \times n_r$ MIMO channel.\\
    \item For $k_t \leq n_t \leq k_r \leq n_r$:
        \begin{align*}
                \C^\star_{k_t,k_r} \stackrel{(e)}{\geq}&\ \frac{k_t}{n_t}\ \C^\star_{n_t,k_r}\\
                \stackrel{(f)}\geq& \frac{k_t}{n_t}\ \C_{n_t,n_r} - \log\left(\frac{ {n_r \choose k_r } }{ {n_r-n_t \choose k_r-n_t } } \right)\\
                \geq& \frac{k_t}{n_t}\ \C_{n_t,n_r} - \log\left( {n_t \choose k_t }  \right) -\log\left( {n_r \choose k_r } \right),
        \end{align*}
        where (e) follows by applying Lemma~\ref{zero_gap_thm} to select an $k_t \times k_r$ subchannel from the $n_t \times k_r$ MIMO channel; The relation (f) follows from \eqref{k_r_larger_n_t}.
\end{enumerate}
By combining the aforementioned cases, we have
\begin{align*}
    \C^\star_{k_t,k_r} \geq \frac{\min{(k_t,k_r)}}{\min{(n_t,n_r)}}\ \C_{n_t,n_r} - \log\left( {n_t \choose k_t }  {n_r \choose k_r } \right),
\end{align*}
which is the lower bound stated in Theorem \ref{thm:mimo_theorem}.

\subsection{Tight Example}
To prove that there exists a class of MIMO channels for which the lower bound in Theorem~\ref{thm:mimo_theorem} is tight (to within a constant gap), consider the $n \times n$ parallel MIMO channel with unit capacities between the $i$-th transmitter and $i$-th receiver.
The capacity $\C_{n,n}$ of this channel is $n$.
For any $(k_t,k_r)$, it is not hard to see that a $k_t \times k_r$ MIMO subchannel can at most capture $\min(k_t,k_r)$ of the parallel channels. Therefore, we have $\C^\star_{k_t,k_r} = \min(k_t,k_r)$ and as a result
\begin{align*}
    \frac{\C^\star_{k_t,k_r}}{\C_{n_t,n_r}} = \frac{\min(k_t,k_r)}{n}.
\end{align*}
This concludes our proof of Theorem~\ref{thm:mimo_theorem}.
\section{Theorem~\ref{thm:single_path}: A General Guarantee for Routes in $N$-relay Gaussian networks}\label{sec:main_results_general}
In this section, we derive the guarantee on the capacity of the best route in a $N$-relay  network presented in Theorem~\ref{thm:single_path}.
We start off by showing that this guarantee is fundamental as there exists a class of networks for which the guarantee is tight up to a constant gap and then proceed to prove the lower bound in the theorem.

\subsection{Tight Examples}
In this subsection, we prove the existence of $N$-relay  networks where the capacity of each path satisfies \eqref{thm:single_path_tight_bound}.
This is sufficient to prove that the ratio in Theorem \ref{thm:single_path} is tight.
To this end, let $N_f \triangleq \floor{(N-1)/2}$ and consider the following network constructions.
For odd $N$ and $A > 0$, we have
\begin{align}
    \label{eq:single_path_ntwk}
    R_{S\to 1}  &= A,\quad R_{1 \to D} = N^2 A, \nonumber \\
    R_{S\to i} &= N^2A,\quad i \in [2:N_f{+} 1], \nonumber\\
    R_{i\to D} &= N^2A,\quad i \in [N_f{+} 2 :2 N_f {+}1],\\
    R_{i\to j} &= A,\quad i \in [2:N_f {+}1 ],\ j = i {+} N_f \nonumber\\
    R_{i\to j} &= 0, \qquad \text{otherwise}. \nonumber
\end{align}
For an even $N$, we have that $N=  2N_f +2 $.
Therefore, the construction includes an extra relay connected only to $S$ and $D$ as follows,
\begin{align*}
    R_{S\to N} = N^2A,\qquad R_{N \to D} = A.
\end{align*}
Fig.~\ref{fig:example_ntwks_general} illustrates the network structure for odd and even number of relays.
From the structure and the cut illustrated in Fig.~\ref{fig:example_ntwks_general}, it is clear that the approximate capacity is $\widebar{\C} = A(\floor{N/2} + 1)$.
Additionally, from Fig.~\ref{fig:example_ntwks_general} (and the construction in \eqref{eq:single_path_ntwk}), it is clear that any path that connects $S$ and $D$ passes through a link of capacity $A$.
As a result, we have
\[
    \forall\ \text{paths}\ \mcal{P}\ :\ \C_{\mcal{P}} \leq A = \frac{1}{\floor{N/2}+1} \widebar{\C}.
\]
\subsection{Proof of Lower Bound in Theorem \ref{thm:single_path}}
\begin{figure*}[t!]
    \centering
    \begin{subfigure}[t]{0.5\textwidth}
        \centering
        \includegraphics[height=2in]{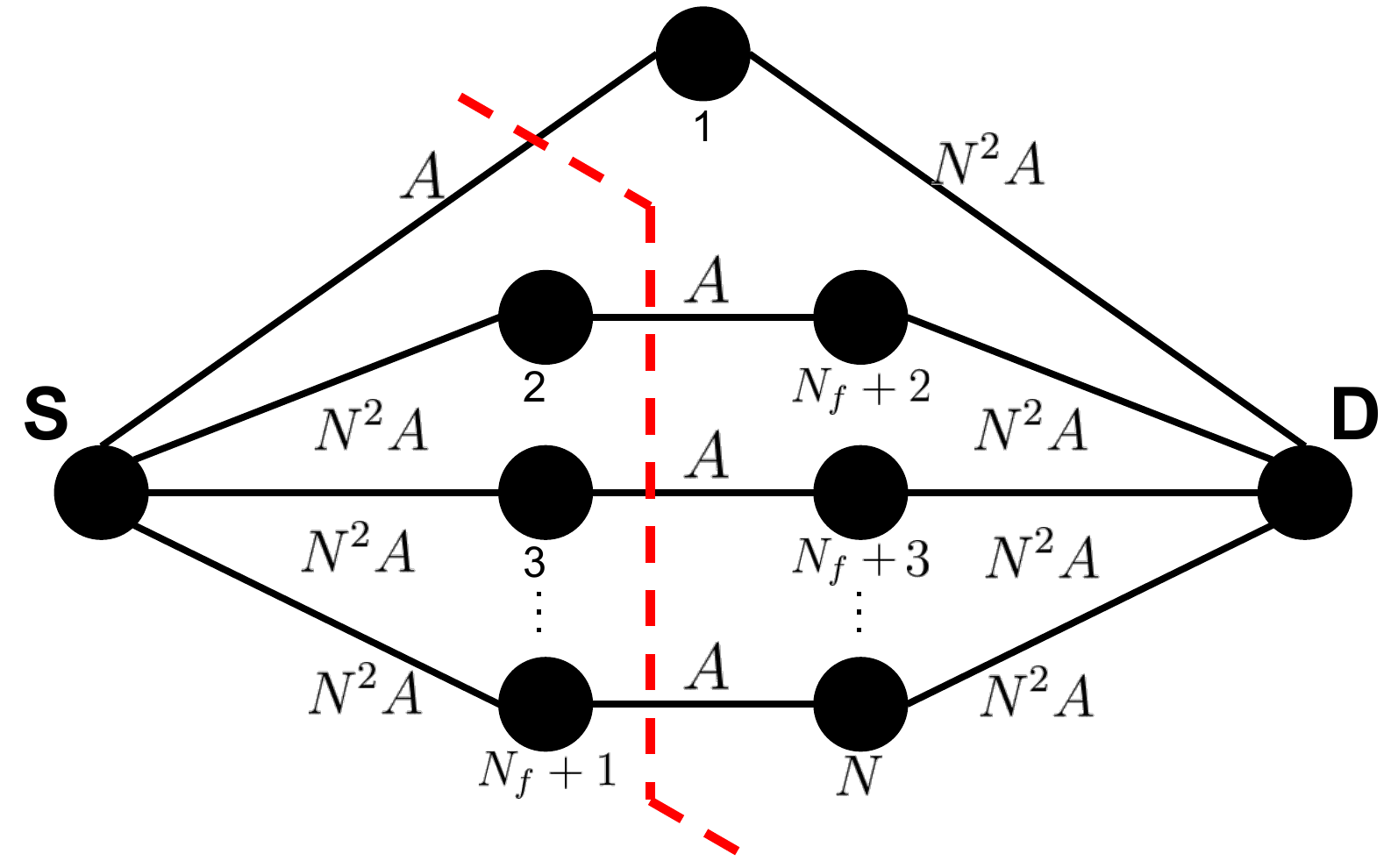}
        \caption{$N$ odd}
    \end{subfigure}%
    ~ 
    \begin{subfigure}[t]{0.5\textwidth}
        \centering
        \includegraphics[height=2.1in]{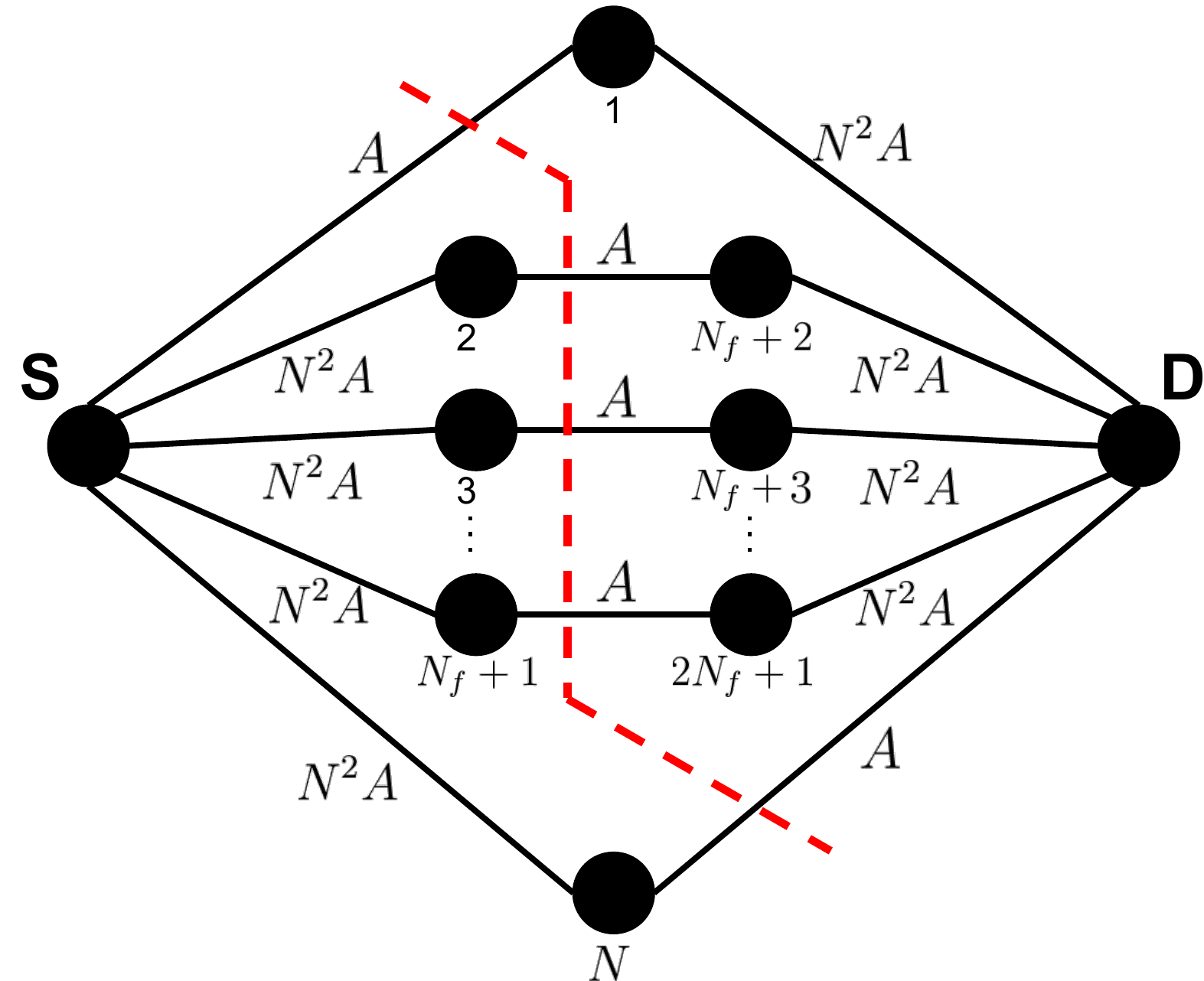}
        \caption{$N$ even}
    \end{subfigure}
    \caption{$N$-relay  networks where every route achieves at most a capacity that is $\frac{1}{1+N/2}$ of the approximate capacity. Edge labels indicate the capacity of the corresponding links. The red line highlights the minimum cut in each network.}
    \label{fig:example_ntwks_general}
\end{figure*}
Next, we prove the lower bound in \eqref{eq:single_path_ntwk}. As discussed in Section \ref{sec:model}, the approximate capacity of an $N$-relay  network is 
\begin{align}
    \widebar{\C} = \min_{\substack{\Omega \in 2^{\mcal{V}}}} \widebar{\C}(\Omega,\mcal{V}),
    \label{eq:C_bar_start}
\end{align}
where $\mcal{V}$ is the set of all nodes in the network.
The term $\widebar{\C}(\Omega,\mcal{V})$ is the capacity of the Gaussian MIMO channel, between nodes $\Omega$ and $\Omega^c$, with i.i.d inputs. 
Therefore, our proof of Theorem \ref{thm:single_path} is built around Theorem~\ref{thm:mimo_theorem} on MIMO channels that was proved in the previous section.

By applying the result in Theorem \ref{thm:mimo_theorem} with $K_t = K_r = 1$ on \eqref{eq:C_bar_start}, we get that
\begin{align*}
    \label{eq:cut_upperbound_general}
    \widebar{\C} &\leq \widebar{\C}(\Omega,\mcal{V}),\qquad& \forall \Omega \subseteq \mcal{V} \\
    &\stackrel{(a)}\leq \min (|\Omega|,|\Omega^c|)\max_{i \in \Omega, j \in \Omega^c} R_{i\to j}\ +\ \min (|\Omega|,|\Omega^c|) \log \left(|\Omega||\Omega^c|\right)
,\qquad& \forall \Omega \subseteq \mcal{V}\\
&\stackrel{(b)}\leq \left(\left\lfloor\frac{N}{2}\right\rfloor + 1\right) \max_{i \in \Omega, j \in \Omega^c} R_{i\to j} + 2\left(\left\lfloor\frac{N}{2}\right\rfloor+1\right)\log\left(\frac{N+2}{2}\right), \qquad& \forall \Omega \subseteq \mcal{V}, \numberthis 
\end{align*}
where: $(a)$ follows from Theorem~\ref{thm:mimo_theorem}; relation $(b)$ follows since $\max_{\Omega}\min(|\Omega|,|\Omega^c|) = \floor{N/2}+1$ and $\max_{\Omega}\log(|\Omega||\Omega^c|) \leq 2\log((N+2)/2)$.
We can now use the upper bound in \eqref{eq:cut_upperbound_general} to prove Theorem~\ref{thm:single_path} by contradiction.

Assume that for all paths $\mcal{P}$ in the network, the capacity of the path $\C_\mcal{P}$ is
\begin{equation}
    \C_\mcal{P} < \frac{1}{\floor{N/2}+1} \widebar{\C} - 2\log\left(\frac{N+2}{2}\right).
    \label{eq:contradiction_bound_general}
\end{equation}
Let $\mcal{B} = \{(i,j)\ |\ R_{i\to j} < \frac{1}{\floor{N/2}+1} \widebar{\C} - 2\log((N+2)/2) \}$ be the set of links that have a capacity strictly less than the bound in Theorem \ref{thm:single_path}.
The relation \eqref{eq:contradiction_bound_general} implies that every path $\mcal{P}$ has a at least one link that belongs to $\mcal{B}$. Therefore, removing $\mcal{B}$ completely disconnects the source and destination.
As a result, we can construct a cut $\Omega_\mcal{B}$ in the network by removing links in $\mcal{B}$. Note that since $\Omega_\mcal{B}$ is constructed using links from $\mcal{B}$, then we have
\begin{equation}
    \max_{i \in \Omega_\mcal{B}, j \in \Omega_\mcal{B}^c} R_{i\to j} < \frac{1}{\floor{N/2}+1} \widebar{\C} -2 \log\left(\frac{N+2}{2}\right).
    \label{eq:contradiction_bound_general_conc}
\end{equation}
If we now apply \eqref{eq:cut_upperbound_general} for $\Omega_\mcal{B}$ we get
\begin{align*}
    \widebar{\C} &\leq \left(\left\lfloor\frac{N}{2}\right\rfloor + 1\right) \max_{i \in \Omega_\mcal{B}, j \in \Omega_\mcal{B}^c} R_{i\to j} + 2\left(\left\lfloor\frac{N}{2}\right\rfloor + 1\right)\log\left(\frac{N+2}{2}\right)\\
    &< \left(\left\lfloor\frac{N}{2}\right\rfloor + 1\right) \left[\frac{1}{\floor{N/2}+1} \widebar{\C} - 2\log\left(\frac{N+2}{2}\right)\right]+ 2\left(\left\lfloor\frac{N}{2}\right\rfloor + 1\right)\log\left(\frac{N+2}{2}\right)\\
    &= \widebar{\C},
\end{align*}
which is a contradiction. This completes our proof for Theorem \ref{thm:single_path}.

\section{Theorem~\ref{thm:single_path_layered}: A Guarantee for routes in $N$-relay Gaussian layered networks}\label{sec:main_results_layered}
This section is dedicated to the proof of Theorem~\ref{thm:single_path_layered}.
Before delving into the proof of the theorem, we update our notation to fit the special class of layered networks.  
\subsection{Updated Notations}
For a $N$-relay Gaussian layered network with $L$ relay layers and $N_L = N/L$ relays per layer, we can decompose the set of nodes $\mcal{V}$ into the subsets $\mcal{V}_l$, $l \in [1:L]$.
$\mcal{V}_l$ is the set of all network nodes that belong to the $l$-th layer. 
Note that $\mcal{V}_0 = \{S\}$ and $\mcal{V}_{L+1}  = \{D\}$.
Similarly, we can decompose the cut $\Omega \subseteq \mcal{V}$ into the disjoint subsets $\Omega^{(l)} = \Omega \cap \mcal{V}_l$.
In a layered network, nodes in layer $l$ receive transmissions only from nodes in the preceding layer ($l-1$).
Therefore, we can use the additional notation to rewrite $\widebar{\C}(\Omega,\mcal{V})$ in \eqref{eq:cut_value} as
\begin{align*}\label{eq:cut_value_layered}
    \widebar{\C}(\Omega,\mcal{V}) &= \log\text{det}\left( \mb{I} + \mb{H}_\Omega {\mb{H}_\Omega}^\dagger \right)\\
    &= \sum_{l=0}^L \underbrace{\log\text{det}\left( \mb{I} + \mb{H}_{\Omega_l} {\mb{H}_{\Omega_l}}^\dagger \right)}_{\widebar{\C}_l(\Omega,\mcal{V})}, \numberthis
\end{align*}
where $\mb{H}_{\Omega_l}$ represents a MIMO channel matrix from nodes in $\Omega_l$ to nodes in $\Omega_{l+1}^c = \mcal{V}_{l+1}\backslash \Omega_{l+1}$.
Additionally, we use $R^{(l)}_{i \to j}$ to denote the capacity of the link connecting the $i-th$ node in layer $l$ to the $j$-th node in the following layer ($l+1$) as follows. 
\[
    R^{(l)}_{i \to j} = R_{\hat{i}\to \hat{j}},\qquad \ \hat{i} \triangleq i + N_L\times (l-1),\ \hat{j} \triangleq j + N_L\times l.
\]
With this additional notation, we now prove Theorem \ref{thm:single_path_layered} in the two following subsections.
    \begin{figure}
        \centering
        \includegraphics[width=0.5\textwidth]{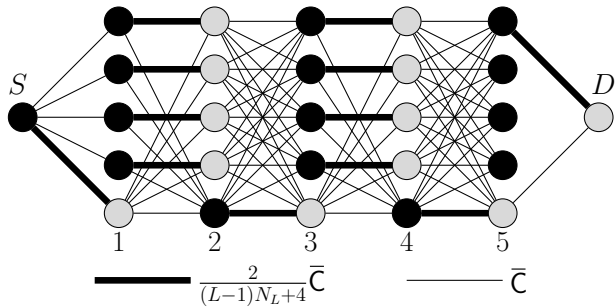}
        \caption{Example network with $N_L=5$ relays per layer and $L=5$ relay layers. Dark nodes represent nodes on the $Source$ side of the cut.}
        \label{fig:converse_odd_layers_relays}
    \end{figure}
\subsection{Tight Examples}
To prove that the bound in Theorem \ref{thm:single_path_layered} is tight (within a constant gap), it suffices to provide example networks where the maximum capacity of any path in the network (by selecting one relay per layer) satisfies \eqref{eq:single_path_layered_ntwk}.
We provide two different constructions based on whether the number of layers $L$ is odd or even.
\begin{itemize}
    \item \underline{Odd $L$}\\
    Consider the example network illustrated in Fig.~\ref{fig:converse_odd_layers_relays} for $L=5$ layers of relays.
    The general construction of the network in Fig.~\ref{fig:converse_odd_layers_relays} for arbitrary odd $L$ and $N_L$ relays per layer is:
    \[
        \begin{aligned}
            &\ R^{(0)}_{S\to i} = R^{(L)}_{N_L\to D}= \widebar{\C}\ \ \quad\quad\quad\quad\quad\quad \forall i \in [1: N_L{-}1] \\
            &\ R^{(0)}_{S\to N_L}= R^{(L)}_{1\to D} = \frac{2}{(L-1)N_L+4} \widebar{\C}\\
            &\ R^{(L)}_{i\to D} = 0\ \ \quad\quad\quad\quad\quad\quad\quad\quad\quad\quad\ \forall i \in [2: N_L{-}1] \\
            l\ \text{odd} & \quad (l \neq L): \\
            &\ R^{(l)}_{i\to i} = \frac{2}{(L-1)N_L + 4} \widebar{\C} \ \ \quad\quad\quad\quad\ \forall i \in [1:N_L{-}1] \\
            &\ R^{(l)}_{i\to j} = 0 \quad\quad\quad\quad \qquad\qquad\qquad \forall i \in [1:N_L{-}1],\ i \neq j \\
            &\ R^{(l)}_{N_L\to i} = R^{(l)}_{i\to N_L} = \widebar{\C}  \qquad\quad\quad\ \ \ \forall i \in [1:N_L] \\
            l\ \text{even} & \quad (l \neq 0):\\
            &\ R^{(l)}_{i\to j} = \widebar{\C} \quad\quad\quad\quad\quad\quad\quad\quad\quad\quad\quad\ \ \forall i,j \in [1:N_L],\ (i,j) \neq (N_L,N_L) \\
            &\ R^{(l)}_{N_L\to N_L} = \frac{2}{(L-1)N_L+4}\widebar{\C}.
        \end{aligned}
    \]

    \begin{figure}
        \centering
        \includegraphics[width=0.7\textwidth]{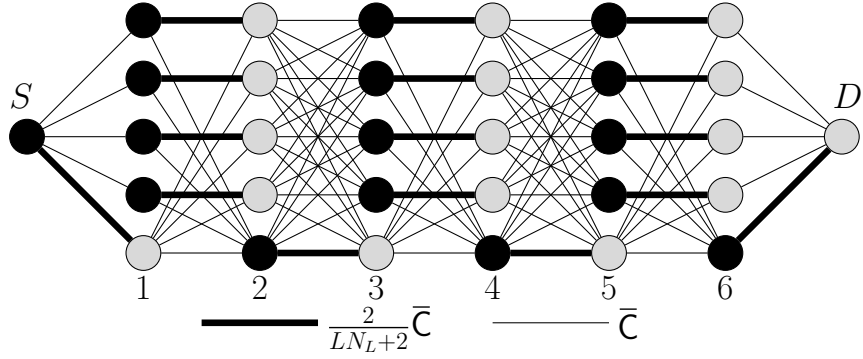}
        \caption{Example network with $N_L=5$ relays per layer and $L=6$ relay layers. Dark nodes represent nodes on the $Source$ side of the cut.}
        \label{fig:converse_even_layers_relays}
    \end{figure}
    It is easy to see that for all cuts except the one highlighted in Fig.~\ref{fig:converse_odd_layers_relays} ($\Omega$ represented by the black nodes), the capacity is greater than or equal $\widebar{\C}$.
    In particular, if any node switches sides (from $\Omega$ to $\Omega^c$ or vice versa), a link of capacity $\widebar{\C}$ would be added to the cut value.
Furthermore, any path from $S$ to $D$ in Fig. \ref{fig:converse_odd_layers_relays} has at least one link with capacity $\frac{2}{(L-1)N_L +4} \widebar{\C}$. Therefore, all routes have capacity of at most $\frac{2}{(L-1)N_L +4} \widebar{\C}$.
    \item \underline{Even $L$}\\
        For even $L$, we consider the network illustrated in Fig.~\ref{fig:converse_even_layers_relays}, which follows the following general construction:
    \[
        \begin{aligned}
            &\ R^{(0)}_{S\to i} = R^{(L)}_{i\to D} = \widebar{\C}\quad\quad\quad\quad\qquad \forall i \in [1: N_L{-}1] \\
            &\ R^{(0)}_{S\to N_L} = R^{(L)}_{N_L\to D} = \dfrac{2}{LN_L+2}\widebar{\C}\\
            l\ \text{odd}& :\\
            &\ R^{(l)}_{i\to i} = \dfrac{2}{LN_L+2}\widebar{\C} \qquad\quad \quad\qquad \forall i \in [1:N_L{-}1] \\
            &\ R^{(l)}_{i\to j} = 0 \quad\quad\quad\quad \qquad\qquad\qquad \forall i \in [1:N_L{-}1],\ i \neq j \\
            &\ R^{(l)}_{N_L\to i} = R^{(l)}_{i\to N_L} = \widebar{\C}  \qquad\quad\quad\ \ \ \forall i \in [1:N_L] \\
            l\ \text{even} & \quad (l \neq 0,\ l \neq L):\\
            &\ R^{(l)}_{i\to j} = \widebar{\C} \quad\quad\quad\quad\quad\quad\quad\quad\quad\quad\quad\ \ \forall i,j \in [1:N_L],\ (i,j) \neq (N_L,N_L) \\
            &\ R^{(l)}_{N_L\to N_L} = \frac{2}{LN_L+2}\widebar{\C}.
        \end{aligned}
    \]
    Similar to the case for odd $L$, the highlighted cut ($\Omega$ represented by the black nodes) is the minimum cut, since it avoids all links with capacity $\widebar{\C}$.
    For the cut illustrated in Fig.~\ref{fig:converse_even_layers_relays}, all paths from $S$ to $D$ include at least one link belonging to the highlighted cut. Therefore any path from $S$ to $D$ has a capacity of at most $\frac{2}{LN_L+2} \widebar{\C}$.

\end{itemize}
\subsection{Proof of Lower Bound in Theorem \ref{thm:single_path_layered}}
The proof of Theorem \ref{thm:single_path_layered} starts by applying the result in Theorem \ref{thm:mimo_theorem} for MIMO channel with i.i.d inputs for $k_t = k_r = 1$.
We apply this on the components of $\widebar{\C}(\Omega,\mcal{V})$ in \eqref{eq:cut_value_layered} as follows
\begin{align*}\label{eq:cut_upperbound_layered}
    \widebar{\C} &\leq \widebar{\C}(\Omega,\mcal{V}) = \sum_{l=0}^L \widebar{\C}_l(\Omega,\mcal{V}) ,\qquad& \forall \Omega \subseteq \mcal{V} \\
    &\stackrel{(a)}\leq \sum_{l=0}^L \left\{\min(|\Omega_l|,|\Omega_{l+1}^c|) \left[ \max_{i \in \Omega_l,\ j \in \Omega_{l+1}^c} R^{(l)}_{i \to j} + \log \left(|\Omega_l||\Omega_{l+1}^c|\right)\right]\right\} ,\qquad& \forall \Omega \subseteq \mcal{V} \\
&\stackrel{(b)}\leq \sum_{l=0}^L \left\{\min(|\Omega_l|,|\Omega_{l+1}^c|) \left[ \max_{i \in \Omega_l,\ j \in \Omega_{l+1}^c} R^{(l)}_{i \to j} + 2\log \left(N_L\right)\right]\right\} ,\qquad& \forall \Omega \subseteq \mcal{V} \\
    &\leq \left[\sum_{l=0}^L \min(|\Omega_l|,|\Omega_{l+1}^c|)\right]\left( \max_{l \in[0:L]} \max_{i \in \Omega_l,\ j \in \Omega_{l+1}^c} R^{(l)}_{i \to j} + 2\log \left(N_L\right)\right) ,\qquad& \forall \Omega \subseteq \mcal{V}\\
&\stackrel{(c)}= \overbrace{\left[\sum_{l=0}^L \min(|\Omega_l|,|\Omega_{l+1}^c|)\right]}^{T(\Omega)}
\left(\max_{i \in \Omega,\ j \in \Omega^c} R_{i \to j} + 2\log \left(N_L\right)\right),\qquad& \forall \Omega \subseteq \mcal{V}.
\numberthis
\end{align*}
where: $(a)$ follows by applying Theorem \ref{thm:mimo_theorem} to each term in the summation; $(b)$ follows from the fact that $|\Omega_l|,\ |\Omega_{l+1}| \leq N_L, \forall l \in [1:L]$;
the relation $(c)$ follows from the fact that $R_{i \to j} = 0$ forall $i,j$ that do not belong to successive layers.

At this point, the expression \eqref{eq:cut_upperbound_layered} looks similar to \eqref{eq:cut_upperbound_general} in our proof for Theorem \ref{thm:single_path}.
Using similar contradiction arguments as in the proof of Theorem \ref{thm:single_path} in the previous section, we can prove that there exists a path $\mcal{P}$ such that
\begin{align}\label{eq:layered_pre_final_step}
    \C_{\mcal{P}} \geq \frac{1}{\displaystyle\max_{\Omega \subseteq \mcal{V}} T(\Omega)}\widebar{\C} - 2 \log(N_L).
\end{align}
Our final step is to get an upper bound on $T(\Omega)$.
This is done through the following Property.
\begin{property}
    For a layered network with $L$ relay layers and $N_L$ relays per layer, define $T(\Omega)$ for a particular cut $\Omega$ as:
    \begin{equation}
        T(\Omega) \triangleq \sum_{l=0}^{L} \min(\vert \Omega_l\vert,\vert \Omega_{l+1}^c \vert).
        \label{eq:additive_terms}
    \end{equation}
    Then we have,
        \begin{align*}
        \max_{\Omega \subseteq \mcal{V}} T(\Omega) \leq T_{max}(L) =
        \begin{cases}
            \frac{(L-1)N_L +4}{2}, & \quad L\ \text{odd} \\
            \frac{LN_L + 2}{2} , & \quad L\ \text{even}
            \end{cases}
        \end{align*}
 \label{prpty:minimalcuts_N2}
\end{property}
\begin{proof}
        Assuming $L$ is odd, we can rewrite \eqref{eq:additive_terms} as
        \begin{align*}
        T(\Omega) = \sum_{l=0}^L \min (|\Omega_l|,|\Omega_{l+1}^c|) &= \min (|\Omega_0|,|\Omega_{1}^c|) + \sum_{l=1}^{L-1} \min (|\Omega_l|,|\Omega_{l+1}^c|) + \min (|\Omega_L|,|\Omega_{L+1}^c|) \\
        &\leq 2 +\sum_{l=1}^{\frac{L-1}{2}} \min (|\Omega_{2l}|,|\Omega_{2l+1}^c|) + \sum_{l=1}^{\frac{L-1}{2}} \min (|\Omega_{2l+1}|,|\Omega_{2l+2}^c|)\\
        &\leq 2 +\sum_{l=1}^{\frac{L-1}{2}} |\Omega_{2l+1}^c| + \sum_{l=1}^{\frac{L-1}{2}} |\Omega_{2l+1}| = 2 + \frac{(L-1)N_L}{2}.
        \end{align*}

        For an even $L$, the result follows similarly as follows
        \begin{align*}
        T(\Omega) = \sum_{l=0}^L \min (|\Omega_l|,|\Omega_{l+1}^c|) &= \min (|\Omega_0|,|\Omega_{1}^c|) + \sum_{l=1}^{L} \min (|\Omega_l|,|\Omega_{l+1}^c|) \\
        &\leq 1 +\sum_{l=1}^{\frac{L}{2}} \min (|\Omega_{2l}|,|\Omega_{2l+1}^c|) + \sum_{l=1}^{\frac{L}{2}} \min (|\Omega_{2l+1}|,|\Omega_{2l+2}^c|)\\
        &\leq 1 +\sum_{l=1}^{\frac{L}{2}} |\Omega_{2l+1}^c| + \sum_{l=1}^{\frac{L}{2}} |\Omega_{2l+1}| = 1 + \frac{LN_L}{2}.
        \end{align*}
\end{proof}
By using Property~\ref{prpty:minimalcuts_N2} on \eqref{eq:layered_pre_final_step}, we can get
        \begin{align}
            \C_\mcal{P} \geq \begin{cases}
                \dfrac{2}{(L-1)N_L + 4} \widebar{\C} - 2\log(N_L)  , & \quad L\ \text{odd} \\
                \dfrac{2}{LN_L+2} \widebar{\C} - 2\log(N_L) , & \quad L\ \text{even}.
            \end{cases}
        \end{align}
        which is exactly the lower bound in Theorem~\ref{thm:single_path_layered}.
        This concludes our proof.

\section{Conclusion}\label{sec:conclusion}
In this paper we have proved that in a wireless network with $N$ relays, multi-hop routing along the best path can only guarantee a fraction of the approximate network capacity that scales inverse linearly with the number of nodes in the network -  within a constant additive gap that depends only on the number of relays.
    This is a surprising result, as initial results~\cite{NOF_Simplification} showed that for diamond networks (where the relays can only communicate with the source and destination but not among themselves), the guarantee on the capacity of the best route is a constant fraction (1/2) of the approximate capacity of the full network - with a constant additive gap depending on the number of nodes $N$, i.e., as the number of nodes increases, the multiplicative factor remains constant while the additive factor changes.
    Here, we show that for a network that allows communication among the relays, the capacity achieved by a route (although increases with the number of relays in the network due to path-diversity) is only guaranteed a fraction of the physical layer cooperation approximate capacity that is not longer constant but decreases as the number of relays $N$ increases (i.e., in this case, the multiplicative factor drops with $N$ as well).

The key idea in our approach was to view the minimum cut in a network as a MIMO channel where a route between the source and destination amounts to selecting a single link in this MIMO channel.
Based on the MIMO subchannel selection result in Theorem \ref{thm:mimo_theorem}, we characterized the guarantee on the capacity of a multi-hop route in the network in terms of the dimensions of the minimum cut.


\begin{appendices}
    \section{Proof of Lemma \ref{lem:incremental_mimo}}
    \label{appendix:mimo_proof}
    Throughout this section, for any polynomial $g(x)$, we define $[x^j]g(x)$ to be the coefficient of $x^j$ in the polynomial $g(x)$.
    To prove Lemma~\ref{lem:incremental_mimo}, our arguments use properties of principal submatrices of Hermitian matrices, most notably the following. 
\begin{property}\label{property_submatrices}
    Let $\mb{A}$ be an $n \times n$ Hermitian matrix.
    For a subset $\Lambda \subseteq [1:n]$, define $\mb{A}_\Lambda$ to be a principal submatrix of $\mb{A}$, constructed only from the rows and columns of $\mb{A}$ indexed by $\Lambda$.
    Denote with $\rho(\lambda)$ and $\rho_\Lambda(\lambda)$ the characteristic polynomials of $\mb{A}$ and $\mb{A}_\Lambda$, respectively.
    Then the following property holds:
    \begin{equation}
        \label{rho_2_rhosub_gn_char}
        (n-k)!\sum_{\substack{\Lambda \subseteq [1:n],\\|\Lambda|=k}}\rho_{\Lambda}(\lambda)= \rho^{(n-k)}(\lambda),
    \end{equation}
    where: (i) the summations in \eqref{rho_2_rhosub_gn_char} are over all subsets of $[1:n]$ of cardinality $k$;
    (ii) $f^{(j)}(x)$ is the $j$-th derivative of $f(x)$ with respect to $x$.
\end{property}
Property \ref{property_submatrices} is mentioned in \cite{thompson} as a well-known fact.
For completeness, we include a proof of the property at the end of this appendix.
The proof of the property is based on the multilinearity of the determinant of a matrix in terms of its rows.
This property provides us with a key ingredient to the proof of Theorem \ref{thm:mimo_theorem}.
In particular, we are interested in comparing coefficients of $\lambda^0$ in \eqref{rho_2_rhosub_gn_char}.
Let $\{\lambda_i\}_{i=1}^n$ be the set of eigenvalues of the Hermitian matrix $\mb{A}$. 
Then by comparing the coefficients of $\lambda^0$ in \eqref{rho_2_rhosub_gn_char}, we get
\begin{align}
    \label{rho_2_rhosub_gn}
    & [\lambda^0] \left( (n-k)!\sum_{\substack{\Lambda \subseteq [1:n],\\|\Lambda|=k}}\rho_{\Lambda}(\lambda) \right) = [\lambda^0]\rho^{(n-k)}(\lambda) \nonumber \\
    \implies\quad & (n-k)!\sum_{\substack{\Lambda \subseteq [1:n],\\|\Lambda|=k}}[\lambda^0] \rho_{\Lambda}(\lambda) = [\lambda^0]\rho^{(n-k)}(\lambda) \nonumber\\
    \implies\quad & (n-k)!\sum_{\substack{\Lambda \subseteq [1:n],\\|\Lambda|=k}}[\lambda^0] \rho_{\Lambda}(\lambda) =(n-k)![\lambda^{n-k}]\rho(\lambda) \nonumber\\
    \stackrel{(a)}\implies\quad & (n-k)!\sum_{\substack{\Lambda \subseteq [1:n],\\|\Lambda|=k}} (-1)^k \left| [\lambda^0] \rho_{\Lambda}(\lambda)\right| =(n-k)!  (-1)^k \sum_{\substack{\Lambda \subseteq [1:n],\\|\Lambda| = k }} \prod_{j \in \Lambda} \lambda_{j} \nonumber\\
    \implies\quad &\ \quad\qquad \sum_{\substack{\Lambda \subseteq [1:n],\\|\Lambda|=k}}\left| [\lambda^0] \rho_{\Lambda}(\lambda)\right| =  \sum_{\substack{\Lambda \subseteq [1:n],\\|\Lambda| = k }} \prod_{j \in \Lambda} \lambda_{j},
\end{align}
where in $(a)$ the RHS follows from the fact that $\rho(\lambda) = \prod_{i=1}^n(\lambda - \lambda_i)$.
Relation \eqref{rho_2_rhosub_gn} will be the main ingredient in the proofs of Lemma~\ref{lem:incremental_mimo}
as we see in the following.

\medskip

For the channel matrix $\mb{H} \in \mbb{C}^{n_r \times n_t}$ where $n_t \leq n_r$, let $\mathbf{F} = \mathbf{I} + \mathbf{H}\mathbf{H}^\dagger$ and define $\lambda_1 \geq \lambda_2 \geq \dots \geq \lambda_{n_r}$ to be the eigenvalues of $\mathbf{F}$.
Without loss of generality, we assume that each transmitter transmits unity power.
This is because we can rewrite $\C_{n_t,n_r} = \log\det(\mathbf{I} + P\mathbf{H}\mathbf{H}^\dagger)$ as $\log\det(\mathbf{I} + \tilde{\mathbf{H}}\tilde{\mathbf{H}}^\dagger)$ where $\tilde{\mathbf{H}} = \sqrt{P}\mathbf{H}$.
Since we have $n_t \leq n_r$, there exists at most $n_t$ eigenvalues of $\mathbf{F} = \mathbf{I} + \mathbf{H}\mathbf{H}^\dagger$ that are not equal to unity, i.e.,  $\lambda_i = 1$ for $i \in [ n_t+1: n_r]$.
We now appeal to Property \ref{property_submatrices} of characteristic polynomials of submatrices.
 Let $\rho(\lambda)$ and $\rho_\Lambda(\lambda)$ be the characteristic polynomials of $\mb{F}$ and $\mb{F}_\Lambda$, respectively. Here $\mb{F}_\Lambda$ denotes the submatrix of $\mb{F}$ constructed from the rows and columns indexed by $\Lambda$.
For our purposes, $n$ and $k$ in \eqref{rho_2_rhosub_gn} are replaced with $n_r$ and $k_r$, respectively to give the following:
\begin{align}
    \sum_{\substack{\Lambda \subseteq [1:n_r],\\|\Lambda|=k_r}}|[\lambda^0]\rho_{\Lambda}(\lambda)| = \sum_{\substack{\Lambda \subseteq [1:n_r],\\|\Lambda| = k_r }} \prod_{j \in \Lambda} \lambda_{j}.
    \label{eq:relation_edited}
\end{align}
Using \eqref{eq:relation_edited}, we can now prove the two cases in Lemma~\ref{lem:incremental_mimo}.

\noindent \textit{\textbf{Case 1: ($k_t = n_t,\ k_r \leq n_t \leq n_r$)}}\\
The expression in \eqref{eq:relation_edited} can be simplified when $k_r \leq n_t$ as follows:
\begin{align}\label{const_gap_case_1}
    \sum_{\substack{\Lambda \subseteq [1:n_r],\\|\Lambda|=k_r}}\left| [\lambda^0]\rho_{\Lambda}(\lambda)\right| &= \sum_{\substack{\Lambda \subseteq [1:n_r],\\|\Lambda| = k_r }} \prod_{j \in \Lambda} \lambda_{j} \nonumber\\
    \stackrel{(a)}{\geq}& {n_t \choose k_r} \sum_{\substack{\Lambda \subseteq [1:n_t],\\|\Lambda| = k_r }} \frac{1}{{n_t \choose k_r}}\prod_{j \in \Lambda} \lambda_{j} \nonumber \\
    \stackrel{(b)}{\geq}& {n_t \choose k_r} \prod_{\substack{\Lambda \subseteq [1:n_t],\\|\Lambda| = k_r }} \left( \prod_{j \in \Lambda} \lambda_{j} \right)^{{n_t \choose k_r}^{-1}}\nonumber\\
    =& {n_t \choose k_r}\left(\prod_{i=1}^{n_t}\lambda_{i}\right)^{{n_t-1 \choose k_r-1}{n_t \choose k_r}^{-1}}\nonumber\\
    \stackrel{(c)}{=}& {n_t \choose k_r}\left(\prod_{i=1}^{n_r}\lambda_{i}\right)^{\frac{k_r}{n_t}},
\end{align}
where: (a) follows by considering only $k_r$-tuples of the eigenvalues $\lambda_i$ where $i \in [1:n_t]$ -
since $[1:n_t] \subseteq [1:n_r]$, then all $k_r$-tuples from $[1:n_t]$ are contained within the summation in \eqref{eq:relation_edited} and therefore the relation follows;
relation (b) follows from the AM-GM inequality; 
(c) follows by the simplification of the exponent and the fact that $\lambda_i = 1$ for $i \in \{n_t+1,\dots n_r\}$.

By averaging the left hand side of \eqref{const_gap_case_1}, we have
\begin{align*}
    \dfrac{1}{{n_r \choose k_r}}  \sum_{\substack{\Lambda \subseteq [1:n_r],\\|\Lambda|=k_r}}\left| [\lambda^0]\rho_{\Lambda}(\lambda) \right| \geq \frac{{n_t \choose k_r}}{{n_r \choose k_r}} \left(\mathlarger{\prod}_{i=1}^{n_r}\lambda_{i}\right)^{\frac{k_r}{n_t}},
\end{align*}
This implies that there exists some selection $\Lambda_s$ of $k_r$ receivers such that if we take $\mb{H}_{\Lambda_s}$, the submatrix of $\mb{H}$ that selects only rows indexed by $\Lambda_s$, then the matrix $\mathbf{B}_s = \mathbf{I} + \mathbf{H}_{\Lambda_s}{\mathbf{H}_{\Lambda_s}}^\dagger$ satisfies
\[
    \begin{aligned}
        \log\det(\mathbf{B}_s) = \log \left(\left| [\lambda^0]\rho_{\Lambda_s}(\lambda)\right| \right) \geq& \log \left(\frac{{n_t \choose k_r}}{{n_r \choose k_r}} \left(\mathlarger{\prod}_{i=1}^{n_r}\lambda_{i}\right)^{\frac{k_r}{n_t}} \right) \\
        =& \frac{k_r}{n_t}\log\det(\mathbf{F}) - \log\left(\frac{{n_r \choose k_r}}{{n_t \choose k_r}}\right).
    \end{aligned}
\]
Since $\C_{n_t,n_r} = \log\det\left(\mb{I} + \mb{H}\mb{H}^\dagger\right) = \log\det\left(\mb{F}\right)$ and $\C^\star_{n_t,k_r} \geq \log\det(\mb{B}_s)$, then we have that
\[
    \C^\star_{n_t,k_r} \geq \frac{k_r}{n_t} \C_{n_t,n_r} - \log\left( \frac{{n_r \choose k_r}}{{n_t \choose k_r}}  \right).
\]
\noindent \textit{\textbf{Case 2: (}$k_t = n_t$, $n_t \leq k_r \leq n_r$\textbf{)}}

Since $k_r \geq n_t$, then there exist sets $\Lambda \subseteq [1:n_r]$ with cardinality $k_r$ such that 
$[1:n_t] \subseteq \Lambda \subseteq [1:n_r]$.
Therefore, we can get a lower bound from \eqref{eq:relation_edited}  as
\begin{align} \label{const_gap_case_2_LB}
    \dfrac{1}{{n_r \choose k_r}}\sum_{\substack{\Lambda \subseteq [1:n_r],\\|\Lambda|=k_r}}\left| [\lambda^0]\rho_{\Lambda}(\lambda)\right| &= \dfrac{1}{{n_r \choose k_r}}\sum_{\substack{\Lambda \subseteq [1:n_r],\\|\Lambda| = k_r }} \prod_{j \in \Lambda} \lambda_{j}\nonumber\\
    &\geq \dfrac{1}{{n_r \choose k_r}} \sum_{\substack{[1:n_t] \subseteq \Lambda \subseteq [1:n_r],\\|\Lambda| = k_r }}\hspace{0.1in} \prod_{j \in \Lambda} \lambda_{j}\nonumber\\
&\stackrel{(a)}{\geq} \dfrac{1}{{n_r \choose k_r}} \sum_{\substack{[1:n_t] \subseteq \Lambda \subseteq [1:n_r],\\|\Lambda| = k_r }}\hspace{0.1in} \prod_{j=1}^{n_t} \lambda_{j} \nonumber\\
&\stackrel{(b)}{=} \dfrac{{n_r-n_t \choose k_r-n_t}}{{n_r \choose k_r}} \left(\prod_{j=1}^{n_t} \lambda_j \right) = \dfrac{{n_r-n_t \choose k_r-n_t}}{{n_r \choose k_r}} \left(\prod_{j=1}^{n_r} \lambda_j \right),
\end{align}
where: (a) follows since $\forall i \in [1:n_r]$, we have $\lambda_i \geq 1$;
relation (b) follows since there are are ${n_r-n_t \choose k_r-n_t}$ sets $\Lambda$ such that $|\Lambda| = k_r$ and $[1:n_t] \subseteq \Lambda \subseteq [1:n_r]$.

The average relation in  \eqref{const_gap_case_2_LB} implies that there exists a selection $\Lambda_s$ of $k_r$-receivers such that if we take $\mb{H}_{\Lambda_s}$, the submatrix of $\mb{H}$ that selects only rows indexed by $\Lambda_s$, then the matrix $\mathbf{B}_s = \mathbf{I} + \mathbf{H}_{\Lambda_s}{\mathbf{H}_{\Lambda_s}}^\dagger$ satisfies
\[
    \log\det(\mathbf{B}_s) = \log \left(\left| [\lambda^0]\rho_{\Lambda_s}(\lambda)\right| \right) \geq \log\det(\mathbf{F}) - \log\left(\frac{{n_r \choose k_r}}{{n_r-n_t \choose k_r-n_t}}\right).
\]
Since $\C_{n_t,n_r} = \log\det\left(\mb{I} + \mb{H}\mb{H}^\dagger\right) = \log\det\left(\mb{F}\right)$ and $\C^\star_{n_t,k_r} \geq \log\det(\mb{B}_s)$, then we have that
, then by choosing $k_r$ receivers we have 
\[
    \C^\star_{n_t,k_r} \geq \C_{n_t,n_r} - \log\left( \frac{{n_r \choose k_r}}{{n_r-n_t \choose k_r-n_t}}  \right).
\]
However, fundamentally we know that $\C^\star_{n_t,k_r} \leq \C_{n_t,n_r}$, therefore, we have:
\[
    \C_{n_t,n_r} \geq \C^\star_{n_t,k_r} \geq  \C_{n_t,n_r} - \log\left( \frac{{n_r \choose k_r}}{{n_r-n_t \choose k_r-n_t}}  \right).
\]

\noindent This concludes the proof of Lemma \ref{lem:incremental_mimo}.

As seen above, our proof relied heavily on Property~\ref{property_submatrices} which we prove next.

\subsection{Proof of Property \ref{property_submatrices}}
\noindent Let $\rho(\lambda)$ denote the characteristic polynomial of the Hermitian matrix $\mathbf{A} \in \mbb{C}^{n \times n}$.
The characteristic polynomial $\rho(\lambda)$ is the determinant of the matrix $(\lambda \mathbf{I} - \mathbf{A})$ and is therefore, 
multilinear in the rows of the matrix $\lambda \mathbf{I} - \mathbf{A}$.
Thus we can write $\rho(\Lambda)$ as 
\[
    \rho(\lambda) = M\big(r_1(\lambda),r_2(\lambda), ... r_n(\lambda)\big),
\]
where $M: \mathbb{C}^n \times \mathbb{C}^n \cdots \times \mathbb{C}^n \rightarrow \mathbb{R}$ is a multilinear mapping and $r_i(\lambda)$ is the $i$-th row of $\mathbf{A}$.
Note that for any multilinear function $M$, the total derivative is the sum of its partial derivatives, i.e.,
\[
    \mbb{D} M(x_1,x_2,..,x_n)(y_1,y_2,..,y_n) = \sum_{i=1}^n M(x_1,...,y_i,..,x_n).
\]
Therefore, by applying the chain rule, we have
\[
    \begin{aligned}
        \rho^{(1)}(\lambda) &=  \mbb{D}M(r_1(\lambda),r_2(\lambda),\cdots,r_n(\lambda))(r_1^{(1)}(\lambda),r_2^{(1)}(\lambda),..,r_n^{(1)}(\lambda))\\
        &= \sum_{i=1}^n M\big(r_1(\lambda), ... , r^{(1)}_i(\lambda), ... r_n(\lambda)\big),
    \end{aligned}
\]
where $r^{(1)}_i(\lambda)$ is the differentiation of the $i$-th row of $\mathbf{A}$ with respect to $\lambda$.
Note that $r^{(1)}(\lambda) = 0$ at all non-diagonal positions and equals 1 at the diagonal position.
Thus, $M\big(r_1(\lambda), ... , r^{(1)}_i(\lambda), ... r_n(\lambda)\big)$ is the determinant of the matrix $\lambda \mathbf{I} - \mathbf{A}$ after replacing the $i$-th row by $r'_i(\lambda)$.
Expanding the determinant along the $i$-th row of this new matrix, we get that
\[
    M\big(r_1(\lambda), ... , r^{(1)}_i(\lambda), ... r_n(\lambda)\big) = 1\times(\lambda \mathbf{I}-\mathbf{A})_{ii},
\]
where $(\lambda \mathbf{I}-\mathbf{A})_{ii}$ is the minor of $\lambda \mathbf{I}-\mathbf{A}$ formed by removing the $i$-th row and $i$-th column, which is equal to $\det(\lambda \mathbf{I} + \mathbf{A}_{[1:n]\setminus i})$.
$\mathbf{A}_{[1:n]\setminus i}$ is the submatrix of $\mathbf{A}$ by removing the $i$-th row and $i$-th column.
As a result, we have that
\begin{equation}
    \rho^{(1)}(\lambda) = \sum_{i=1}^n \det(\lambda \mathbf{I} + \mathbf{A}_{[1:n]\setminus i}) = \sum_{i=1}^n \rho_{[1:n]\setminus i}(\lambda),
    \label{rho_2_rhosub}
\end{equation}
where $\rho_{[1:n]\setminus i}(\lambda)$ denotes the characteristic polynomial of $\mathbf{A}_{[1:n]\setminus i}$ and $i \in \{1,2,\cdots n\}$.
\\
We can now use \eqref{rho_2_rhosub} in addition to an induction argument to prove the relation in \eqref{rho_2_rhosub_gn}.
Let $g_{k+1}(\lambda)$ be the sum of all characteristic equations of $k+1 \times k+1$ principal submatrices of $\mb{A}$, i.e.,
\[
    g_{k+1}(\lambda) = \sum_{\substack{\Lambda \subseteq [1:n],\\|\Lambda|=k+1}} \rho_{\Lambda}(\lambda).
\]
Taking the derivative of $g_{k+1}(\lambda)$ and applying \eqref{rho_2_rhosub}, we get
\begin{equation}
    g^{(1)}_{k+1}(\lambda) =  \sum_ {\substack{\Lambda \subseteq [1:n],\\|\Lambda|=k+1}} \rho^{(1)}_{\Lambda}(\lambda)
    = \sum_{\substack{\Lambda \subseteq [1:n],\\|\Lambda|=k+1}} \sum_{j \in \Lambda} \rho_{\Lambda\setminus j}(\lambda),
    \label{diff_k_plus_1}
\end{equation}
where $\rho_{\Lambda\setminus j}(\lambda)$ is the characteristic polynomial of the $k\times k$ principal submatrix of $\mathbf{A}$ with rows and columns in $\Lambda\setminus j$.
Note that the summation on the RHS of \eqref{diff_k_plus_1} contains ${n \choose k+1}(k+1)$ terms.
Since there are only ${n \choose k}$ submatrices of size $k \times k$, the summation in \eqref{diff_k_plus_1} is bound to have repeated terms. 
It is not hard to verify that 
\[
    {n \choose k+1}(k+1) = {n \choose k}(n-k).
\]
As a result, from the symmetry of the summation in \eqref{diff_k_plus_1}, we can rewrite the expression as
\begin{equation}
    \sum_ {\substack{\Lambda \subseteq [1:n]\\|\Lambda|=k+1}} \rho^{(1)}_{\Lambda}(\lambda) = (n-k) \sum_ {\substack{\Lambda \subseteq [1:n]\\|\Lambda|=k}} \rho_{\Lambda}(\lambda),
    \label{induction_step}
\end{equation}
Equation \eqref{induction_step} serves as our induction hypothesis.
Our base case is what we proved in \eqref{rho_2_rhosub} which can be deduced from \eqref{induction_step} by choosing $k = n-1$.
Therefore, by applying the induction step \eqref{induction_step} at each step, we get
\begin{align*}
   & \rho^{(1)}_{\Lambda}(\lambda) = \sum_ {\substack{\Lambda \subseteq [1:n]\\|\Lambda|=n-1}} \rho_{\Lambda}(\lambda),\\
   \implies & \rho^{(2)}_{\Lambda}(\lambda) = \sum_{\substack{\Lambda \subseteq [1:n]\\|\Lambda|=n-1}} \rho^{(1)}_{\Lambda}(\lambda) \stackrel{\eqref{induction_step}}= 2 \sum_{\substack{\Lambda \subseteq [1:n]\\|\Lambda|=n-2}} \rho_{\Lambda}(\lambda),\\
   \implies & \rho^{(3)}_{\Lambda}(\lambda) = 2 \sum_{\substack{\Lambda \subseteq [1:n]\\|\Lambda|=n-2}} \rho^{(1)}_{\Lambda}(\lambda) \stackrel{\eqref{induction_step}}=3\cdot 2 \sum_{\substack{\Lambda \subseteq [1:n]\\|\Lambda|=n-3}} \rho_{\Lambda}(\lambda),\\
   & \qquad \qquad \qquad \qquad \vdots \\
   \implies & \rho^{(n-k)}_{\Lambda}(\lambda) =  (n-k)! \sum_{\substack{\Lambda \subseteq [1:n]\\|\Lambda|=k}} \rho_{\Lambda}(\lambda).
\end{align*}
The concludes the proof of Property~\ref{property_submatrices}.

\section{Proof of Lemma \ref{zero_gap_thm}}\label{sec:proof_zero_gap}

To prove the lower bound in Lemma \ref{zero_gap_thm}, it suffices to prove the statement for the following two incremental cases:
        \begin{equation}\label{eq:nogap_case1}
        \text{1) For } k_t = n_t,k_r = n_r-1, \C_{n_t,n_r-1}^\star \geq \frac{n_r-1}{n_r} \C_{n_t,n_r},\qquad\qquad\qquad
        \end{equation}
        \begin{equation}\label{eq:nogap_case2}
        \text{2) For } k_t = n_t-1,k_r = n_r, \C_{n_t-1,n_r}^\star \geq \frac{n_t-1}{n_t} \C_{n_t,n_r}.\qquad\qquad\qquad
        \end{equation}
Using the two statements in \eqref{eq:nogap_case1} and \eqref{eq:nogap_case2}, we can reduce an $n_t \times n_r$ system to a $k_t \times k_r$ system as follows: We first remove one receiver antenna to create an $n_t \times (n_r - 1)$ system such that its capacity $\C_{n_t, n_r - 1}^\star \geq \frac{n_r - 1}{n_r}\C_{n_t,n_r} $. From this (particular) $n_t \times (n_r - 1)$ system, we select an $n_t \times (n_r - 2)$ system such that its capacity $\C_{n_t, n_r - 2}^\star \geq \frac{n_r - 2}{n_r -1}\C_{n_t,n_r-1}^\star$, and so on, till we prune the system down to a $n_t \times k_r$ system. We then repeat the above process for transmitter selection on the $n_t \times k_r$ system to prune it progressively to a $k_t \times k_r$ system with capacity $C_{k_t,k_r}^\star$. The result would then follow as
\begin{align*}
\C^\star_{k_t,k_r} \geq& \quad \dfrac{k_t}{k_t + 1} \C^\star_{k_t+1,k_r}\\
			 \geq& \quad \dfrac{k_t}{k_t + 1} \dfrac{k_t+1}{k_t + 2}\C^\star_{k_t+2,k_r}\\
			 \geq& \quad \dfrac{k_t}{k_t + 1} \dfrac{k_t+1}{k_t + 2} .. \dfrac{n_t-1}{n_t}\ \C^\star_{n_t,k_r}\\
             \geq& \quad \dfrac{k_t}{k_t + 1} \dfrac{k_t+1}{k_t + 2} .. \dfrac{n_t-1}{n_t}\dfrac{k_r}{k_r+1}\ \C^\star_{n_t,k_r+1}\\
             \geq& \quad \dfrac{k_t}{n_t}\dfrac{k_r}{k_r+1} .. \dfrac{n_r-1}{n_r}\ \C_{n_t,n_r}\\			 			 			
             \geq& \quad \dfrac{k_t k_r}{n_t n_r}\ \C_{n_t,n_r}.			
\end{align*}
Without loss of generality, we assume that each transmitter transmits unity power.
This is because we can rewrite $\C_{n_t,n_r} = \log\det(\mathbf{I} + P\mathbf{H}\mathbf{H}^\dagger)$ as $\log\det(\mathbf{I} + \tilde{\mathbf{H}}\tilde{\mathbf{H}}^\dagger)$ where $\tilde{\mathbf{H}} = \sqrt{P}\mathbf{H}$, thus proving the Lemma for $P \neq 1$ is equivalent to proving it for $\tilde{\mathbf{H}}$ instead of $\mathbf{H}$. 
We now prove the two cases in \eqref{eq:nogap_case1} and \eqref{eq:nogap_case2}.

\noindent \textit{\textbf{Case 1: ($k_t = n_t$, $k_r = n_r-1$)}}\\
Let $\mathbf{F} = \mathbf{I}+\mathbf{H}\mathbf{H}^\dagger$ and denote its characteristic polynomial by $\rho(\lambda)$. The capacity can then be written as $\C = \log\det(\mathbf{F})$.
We define $ \mathbf{H}_i$ to be the submatrix of $\mathbf{H}$ constructed by dropping the $i$-th receiver antenna ($i$-th row in $\mathbf{H}$).
Let $\mathbf{B}_i = \mathbf{I} + \mathbf{H}_{i} \mathbf{H}^\dagger_{i}$. Therefore, $\C_i = \log\det(\mathbf{I} + \mathbf{H}_{i} \mathbf{H}^\dagger_{i}) = \log\det(\mathbf{B}_i)$ is the capacity of the MIMO channel with the remaining $n_r{-} 1$ receiver antennas.

Again, we use \eqref{eq:relation_edited}, where we substitute $n = n_r$ and $k = n_r-1$. As a result, we get
\begin{align}
    \frac{1}{n_r} \sum_{\substack{\Lambda \subseteq [1:n_r],\\|\Lambda|=n_r-1}}\left| [\lambda^0] \rho_{\Lambda}(\lambda)\right| &= \frac{1}{n_r} \sum_{\substack{\Lambda \subseteq [1:n_r],\\|\Lambda| = n_r-1 }} \prod_{j \in \Lambda} \lambda_{j}\nonumber\\
    &\stackrel{(a)}\geq  \prod_{\substack{\Lambda \subseteq [1:n_r],\\|\Lambda| = n_r-1 }} \left(\prod_{j \in \Lambda} \lambda_{j}\right)^{\frac{1}{n_r}}\nonumber\\
&=  \left(\prod_{i=1}^{n_r} \lambda\right)^{\frac{n_r-1}{n_r}},
\label{eq:gap_free}
\end{align}
where $(a)$ follows from the AM-GM inequality.
Since the LHS of \eqref{eq:gap_free} is a mean over all $\Lambda \subseteq [1:n_r]$ s.t. $|\Lambda| = n_r -1$, then this implies that there exists $\Lambda_s$ such that
\[
    \left| [\lambda^0] \rho_{\Lambda}(\lambda)\right| \geq \left(\prod_{i=1}^{n_r} \lambda\right)^{\frac{n_r-1}{n_r}}.
\]
Let $s = [1:n_r]\backslash\Lambda_s$, then we have 
\[
    \log\det(\mb{B}_s) = \log \left(\left| [\lambda^0] \rho_{\Lambda}(\lambda)\right|\right) \geq \frac{n_r-1}{n_r}\log\left(\prod_{i=1}^{n_r} \lambda\right) = \log\text{det}(\mb{F}),
\]
which implies that 
\[
    \C^\star_{n_t,n_r-1} \geq  \C_s  \geq \frac{n_r-1}{n_r}\log \det\left(\mathbf{F}\right).
\]
Since $\C_{n_t,n_r} = \log\det\left( \mathbf{F} \right)$, we have
\[
    \C^\star_{n_t,n_r-1} \geq \dfrac{n_r-1}{n_r}\ \C_{n_t,n_r}.
\]
This concludes the proof for the first case.
\\

\noindent \textit{\textbf{Case 2: ($k_t = n_t-1$, $k_r = n_r$)}}\\
To prove this case, we use Sylvester's determinant theorem that states that
\[
    \C_{n_t,n_r} = \log\det(\mathbf{I}_{n_r} + \mathbf{H} \mathbf{H}^{\dagger}) = \log\det(\mathbf{I}_{n_t} + \mathbf{H}^{\dagger}\mathbf{H}).
\]
Let $\hat{\mathbf{F}} = \mathbf{I}_{n_t} + \mathbf{H}^\dagger \mathbf{H}$, and therefore, $\C_{n_t,n_r} = \log\det(\hat{\mathbf{F}})$.
We denote by ${\mathbf{H}^\dagger}_{j}$, the submatrix of $\mathbf{H}^\dagger$ after dropping the $j$-th row.
The capacity of this MIMO subchannel can also be written by Sylvester's theorem as $\C_j = \log\det\left(\mathbf{I}_{n_t} + {\mathbf{H}^\dagger}_{j} ({\mathbf{H}^\dagger}_{j})^\dagger  \right) = \log\det(\hat{\mathbf{B}}_j)$ where $\hat{\mathbf{B}}_j$ is the $(n_t -1) \times (n_t -1)$ matrix constructed from $\hat{\mathbf{F}}$ after removing the $j$-th column and row.
The argument to prove the ratio $\frac{n_t-1}{n_t}$ thus follows similarly with $\hat{\mb{B}}_j$ and $\hat{\mb{F}}$ as in Case 1 with $\mathbf{B}_i$ and $\mathbf{F}$.

This concludes the proof of the lower bound in Lemma~\ref{zero_gap_thm}.

\bigskip

\noindent {\bf Tight Example:}
To prove that the lower bound in Theorem \ref{zero_gap_thm} is tight, consider the $n_t \times n_r$ MIMO channel described by $\mb{H} =\sqrt{P} \mb{O}_{n_r,n_t}$ where $\mb{O}_{n_r,n_t}$ is a $n_r \times n_t$ matrix with all entries equal to unity. It is not hard to see that for the described channel,
\[
    \C = \log\det\left(\mb{I} +P \mb{O}_{n_r,n_t}\mb{O}_{n_r,n_t}^\dagger \right) = \log(1+P n_t n_r).
\]
Similarly for any subchannel of size $k_t \times k_r$, the capacity is $\C_{k_t,k_r} = \log(1+P k_t k_r)$.
Note that for $x \approx 0$, we have $\log(1+x) \approx \frac{1}{\ln(2)}x$. Therefore for $P \approx 0$, we get that $\C_{n_t,n_r} \approx \frac{1}{\ln(2)} P n_t n_r$ and similarly $\C_{k_t,k_r} \approx  \frac{1}{\ln(2)} P k_t k_r$. 
Therefore for $P \approx 0$,
\[
    \frac{\C_{k_t,k_r}}{\C} \approx \frac{k_t k_r}{n_t n_r}.
\]
\noindent This concludes our proof of Lemma \ref{zero_gap_thm}.


\end{appendices}

\bibliographystyle{IEEEtran}
\bibliography{references}

\end{document}